\documentclass[runningheads,a4paper]{llncs}

\usepackage{amssymb}
\usepackage{graphicx}
\usepackage{amsfonts,amsmath,amssymb}
\let\doendproof\endproof
\renewcommand\endproof{~\hfill\qed\doendproof}
\usepackage{algorithmicx}
\usepackage{algorithm}
\usepackage{algpseudocode}

\usepackage{url}
\urldef{\mailsa}\path|{fan, osimpson}@ucsd.edu|
\newcommand{\keywords}[1]{\par\addvspace\baselineskip
\noindent\keywordname\enspace\ignorespaces#1}

\newenvironment{blockquote}{%
   \list{}{\rightmargin\leftmargin}\item\relax
   \itshape
}{\endlist}

\newcommand{\R}{\mathbb{R}}

\renewcommand{\L}{\mathcal{L}}
\newcommand{\lapl}{\mathbf{L}}
\renewcommand{\H}{\mathcal{H}}

\newcommand{\hkpr}{\rho_{t,f}}

\newcommand{\vol}{\textmd{vol}}
\newcommand{\phkpr}{\rho_{t,s}}
\newcommand{\phkprapprox}{\hat{\rho}_{t,s}}

\newcommand{\estphkpr}{DistributedEstimatePHKPR}
\newcommand{\alglocalcluster}{DistributedLocalCluster}

\newcommand{\Kparam}{\frac{\log(\epsilon^{-1})}{\log\log(\epsilon^{-1})}}
\newcommand{\rparam}{\frac{1}{\epsilon^{3}}\log n}
\newcommand{\localclustercomplexity}{\frac{\log(\epsilon^{-1}) \log
n}{\log\log(\epsilon^{-1})} + \frac{1}{\epsilon}\log n}
\newcommand{\rK}{\frac{\log(\epsilon^{-1})\log n}{\epsilon^3
\log\log(\epsilon^{-1})}}
\newcommand{\tparam}{\phi^{-1}\log(\frac{2\sqrt{\varsigma}}{1-\epsilon} +
2\epsilon\sigma)}

\begin{document}
\mainmatter  

\title{Distributed Algorithms for Finding Local Clusters Using Heat Kernel
Pagerank\thanks{An extended abstract appeared in \emph{Proceedings of WAW}
(2015).  Research supported in part by Office of Naval Research (ONR)
N00014-13-1-0257.}}

\titlerunning{Distributed Heat Kernel Local Clusters}

%
%
\author{Fan Chung%
\and Olivia Simpson}
\authorrunning{Chung and Simpson}

\institute{Department of Computer Science and Engineering,\\
University of California, San Diego\\
La Jolla, CA 92093\\
\mailsa}

\toctitle{}
\tocauthor{}
\maketitle

\begin{abstract}
A distributed algorithm performs local computations on pieces of input and
communicates the results through given communication links.  When processing a
massive graph in a distributed algorithm, local outputs must be configured as a
solution to a graph problem without shared memory and with few rounds of
communication.  In this paper we consider the problem of computing a local
cluster in a massive graph in the distributed setting.  Computing local clusters
are of certain application-specific interests, such as detecting communities in
social networks or groups of interacting proteins in biological networks.  When
the graph models the computer network itself, detecting local clusters can help
to prevent communication bottlenecks.  We give a distributed algorithm that
computes a local cluster in time that depends only logarithmically on the size
of the graph in the CONGEST model.  In particular, when the conductance of the
optimal local cluster is known, the algorithm runs in time entirely independent
of the size of the graph and depends only on error bounds for approximation.  We
also show that the local cluster problem can be computed in the $k$-machine
distributed model in sublinear time.  The speedup of our local cluster
algorithms is mainly due to the use of our distributed algorithm for heat kernel
pagerank.

\end{abstract}

\section{Introduction}

Distributed computation is an increasingly important framework as the demand for
fast data analysis grows and data simultaneously becomes too large to fit in
main memory.  As distributed systems for large-scale graph processing such as
Pregel~\cite{malewicz2010pregel}, GraphLab~\cite{yucheng2010graphlab}, and
Google's MapReduce~\cite{dean2004mapreduce} are rapidly developing, there is a
need for both theoretical and practical bounds in adapting classical graph
algorithms to a modern distributed and parallel setting.

A distributed algorithm performs local computations on pieces of input and
communicates the results through given communication links.  When processing a
massive graph in a distributed algorithm, local outputs must be configured
without shared memory and with few rounds of communication.  A central problem
of interest is to compute local clusters in large graphs in a distributed
setting.

Computing local clusters are of certain application-specific interests, such as
detecting communities in social networks~\cite{lldm:localnetwork:08} or groups
of interacting proteins in biological networks~\cite{liao:protein:09}.  When the
graph models the computer network itself, detecting local clusters can help
identify communication bottlenecks, where one set of well-connected nodes is
separated from another by a small number of links.  Further, being able to
identify the clusters quickly prevents bottlenecks from developing as the
network grows.

A local clustering algorithm computes a set of vertices in a graph with a small
Cheeger ratio (or so-called conductance as defined in
Section~\ref{sec:definitions}).  Moreover, we ask that the algorithm use only
local information.  In the static setting, an important consequence of this
locality constraint is running times proportional to the size of the output set,
rather than the entire graph.  In this paper, we present the first algorithms
for computing local clusters in two distributed settings that finish in a
sublinear number of rounds of communication.

A standard technique in local clustering algorithms is the so-called
\emph{sweep} algorithm.  In a sweep, one orders the vertices of a graph
according to some real-valued function defined on the vertex set and then
investigates the cut set induced by each prefix of vertices in the ordering.
The classical method of spectral clustering uses eigenvectors as functions for
the sweep.  For local clustering algorithms, the sweep functions are based on
random
walks~\cite{ls:mixingisoperimetric:90,ls:randomwalks:93,st:graphpartitioning:stoc04}.
In~\cite{acl:prgraphpartition:focs06}, the efficiency of the local clustering
algorithm is due to the use of PageRank vectors as the sweep
functions~\cite{bp:websearchanatomy:isdn98}.  In this paper, the main leverage
in the improved running times of our algorithms is to use the heat kernel
pagerank vector for performing a sweep.  In particular, we are able to exploit
parallelism in our algorithm for computing the heat kernel pagerank and give a
distributed random walk-based procedure which requires fewer rounds of
communication and yet maintains similar approximation guarantees as previous
algorithms.


In Section~\ref{sec:models}, we will describe two distributive models -- the
CONGEST model and the $k$-machine model.  We demonstrate in two different
distributed settings that a heat kernel pagerank distribution can be used to
compute local clusters with Cheeger ratio $O(\sqrt{\phi})$ when the optimal
local cluster has Cheeger ratio $\phi$.  With a fast, parallel algorithm for
approximating the heat kernel pagerank and efficient local computations, our
algorithm works on an $n$-vertex graph in the CONGEST, or standard message
passing, model with high probability in at most
$O\left(\localclustercomplexity\right)$ rounds of communication where $\epsilon$
is an error bound for approximation.  This is an improvement over the previously
best-performing local clustering algorithm in~\cite{dassarma:sparsecut:dc} which
uses a personalized PageRank vector and finishes in $O\left( \frac{1}{\alpha}
\log^2 n + n\log n \right)$ rounds in the CONGEST model for any $0 < \alpha <
1$.  We then extend our results to the $k$-machine model to show that a local
cluster can be computed in $\tilde{O}\left(
\frac{\log(\epsilon^{-1})}{\epsilon^3 k^2 \log\log(\epsilon^{-1})} +
\frac{1}{\epsilon k^2} + \left( \frac{\log(\epsilon^{-1})}{k
\log\log(\epsilon^{-1})} + \frac{1}{k\epsilon}\right)\max\left\{
\frac{1}{\epsilon^{3}}, \Delta \right\} \right)$ rounds, where $\Delta$ is the
maximum degree in the graph, with high probability.

\subsection{Related Work}
The idea of computing local clusters with random walks was introduced by
Lov\'asz and Simonovits in their works analyzing the isoperimetric properties of
random walks on graphs~\cite{ls:mixingisoperimetric:90,ls:randomwalks:93}.
Spielman and Teng~\cite{st:graphpartitioning:stoc04} expanded upon these ideas
and gave the first nearly-linear time algorithm for local clustering, improving
the original framework by sparsifying the graph.  The algorithm
of~\cite{st:graphpartitioning:stoc04} finds a local cluster with Cheeger ratio
$O(\sqrt{\phi}\log^{3/2}n)$ in time $O(m(\log n/\phi)^{O(1)})$, where $m$ is the
number of edges in the graph.  Each of these algorithms uses the distribution of
random walks of length $O(\frac{1}{\phi})$.
Andersen et al.~\cite{acl:prgraphpartition:focs06} give a local clustering
algorithm using the distribution given by a PageRank vector.  Their algorithm
promises a $O(\sqrt{\phi}\log^{1/2}n)$ cluster approximation and runs in time
$O(\frac{m}{\phi}\log^4 m)$.  Orecchia et al. use a variant of heat kernel
random walks in their randomized algorithm for computing a cut in a graph with
prescribed balance constraints~\cite{osv:balsep:11}.  A key subroutine in the
algorithm is a procedure for computing $e^{-A}v$ for a positive semidefinite
matrix $A$ and a unit vector $v$ in time $\tilde{O}(m)$ for graphs on $n$
vertices and $m$ edges.  Indeed, heat kernel has proven to be an efficient and
effective tool for local cluster
detection~\cite{kloster2014heat,cs:hkprcluster:iwoca14}.

Andersen and Peres~\cite{ap:evolving:09} simulate a volume-biased evolving set
process to find sparse cuts.  Their algorithm improves the ratio between the
running time of the algorithm on a given run and the volume of the output set
while maintaining similar approximation guarantees as previous algorithms.
Their algorithm is later improved in~\cite{gt:optimalcluster:12}.  Arora, Rao,
and Vazirani~\cite{arora2009expander} give a $O(\sqrt{\log n})$-approximation
algorithm using semi-definite programming techniques, however it is slower than
algorithms based on spectral methods and random walks.

For distributed algorithms, in~\cite{das2013distributed} fast random walk-based
distributed algorithms are given for estimating mixing time, conductance and the
spectral gap of a network. In~\cite{sarma2013fast}, distributed algorithms are
derived for computing PageRank vectors with $O(\frac{1}{\alpha}\log n)$ rounds
for any $0 < \alpha < 1$ with high probability.  Das Sarma et
al.~\cite{dassarma:sparsecut:dc} give two algorithms for computing sparse cuts
in the CONGEST distributed model. The first algorithm uses  random walks and is
based on the analysis of~\cite{st:graphpartitioning:stoc04}.  By incorporating
the results of~\cite{das2013distributed},  they show that the stationary
distribution of a random walk of length $l$ can be computed in $O(l)$ rounds.
The second algorithm in~\cite{dassarma:sparsecut:dc} uses PageRank vectors and
is based on the analysis of~\cite{acl:prgraphpartition:focs06}.  By using the
results of~\cite{sarma2013fast}, the authors of~\cite{dassarma:sparsecut:dc}
compute local clusters in $O((\frac{1}{\phi} + n) \log n)$ rounds with standard
random walks and $O(\frac{1}{\alpha} \log^2 n + n\log n)$ rounds using PageRank
vectors.  

\section{The Setting and Our Contributions}

\subsection{Models of Computation}\label{sec:models}

We consider two models of distributed computation -- the CONGEST model and the
$k$-machine model.  In each, data is distributed across nodes (machines) of a
network which may communicate over specified communication links in rounds.
Memory is decentralized, and the goal is to minimize the running time by
minimizing the number of rounds required for computation for an arbitrary input
graph $G$.  We emphasize that local communication is taken to be free.

\subsubsection{The CONGEST model} The first model we consider is the CONGEST
model. In this model, the communication links are exactly the edges of the input
graph and each vertex is mapped to a dedicated machine.  The CONGEST (or
standard message-passing) model was introduced
in~\cite{pandurangan:congest,peleg:congest} to simulate real-world bandwidth
restrictions across a network.

Due to how the vertices are distributed in the network, we simplify the model by
assuming the computer network is the input graph $G = (V,E)$ on $n = |V|$ nodes
or machines and $m = |E|$ edges or communication links.  Each node has a unique
$\log n$-bit ID.  Initially each node only possesses its own ID and the IDs of
each of its neighbors, and in some instances we may allow nodes some metadata
about the graph (the value of $n$, for instance).  Nodes can only communicate
through edges of the network and communication occurs in rounds.  That is, any
message sent at the beginning of round $r$ is fully transmitted and received by
the end of round $r$.  We assume that all nodes run with the same processing
speed.  Most importantly, we only allow $O(\log n)$ bits to be transmitted
across any edge per round.

\subsubsection{The $k$-machine model} The defining difference between the
$k$-machine model and the CONGEST model is that, whereas vertices are mapped to
distinct, dedicated machines in the CONGEST model, a number of vertices may be
mapped to the same machine in the $k$-machine model.  This model is meant to
more accurately simulate distributed graph computation in systems such as
Pregel~\cite{malewicz2010pregel} and GraphLab~\cite{yucheng2010graphlab}.

We consider computing over massive datasets distributed over nodes of the
$k$-machine network.  The complete data is never known by any individual
machine, and there is no shared memory.  Each machine executes an instance of a
distributed algorithm, and the output of each machine is with respect to the
data points it hosts.  A solution to a full problem is then a particular
configuration of the outputs of each of the machines.  The model is discussed in
greater detail in Section~\ref{sec:kmachine}.

\medskip
The two models are limiting and advantageous in different ways, and one is not
inherently better than the other.  For instance, since many vertices are mapped
to a single machine in the $k$-machine model, there is more ``local
information'' available since vertices sharing a machine can communicate for
free.  However, since communication is restricted to the communication links in
the computer network, vertex-vertex communication is somewhat less restrictive
in the CONGEST model since links exactly correspond to edges.  The consequences
of these differences are largely observed in time complexity, and certain graph
problems are more suited to one model than the other.

In this paper we analyze our algorithmic techniques in the CONGEST model, and
then use the Conversion Theorem of~\cite{klauck2015distributed} to give an
efficient probabilistic algorithm in the $k$-machine model for computing local
clusters.

\subsection{Local Clusters and Heat Kernel Pagerank}\label{sec:definitions}
Throughout this paper, we consider a graph $G = (V,E)$ with $n=|V|$ and $m=|E|$
that is connected and undirected.  In this section we give some definitions that
will make our problem statement and results precise.

\subsubsection{Personalized heat kernel pagerank}\label{sec:hkpr} The heat
kernel pagerank is so named for the \emph{heat kernel} of the graph, $\H_t =
e^{-t\L}$, where $\L$ is the normalized graph Laplacian $\L =
D^{-1/2}(D-A)D^{-1/2}$.  Here $D$ is the diagonal matrix whose entries
correspond to vertex degree and $A$ is the symmetric adjacency matrix.  The heat
kernel is a solution to the heat equation $\frac{\partial u}{\partial t} = -\L
u$, and thus has fundamental connections to diffusion properties of a graph.
Because of its connection to random walks, for heat kernel pagerank we use a
similar heat kernel matrix, $H_t = e^{-t\lapl}$, where $\lapl = I - P$.  Here,
$I$ is the $n \times n$ identity matrix and $P = D^{-1}A$ is the transition
probability matrix corresponding to the following standard random walk on the
graph: at each step, move from a vertex $v$ to a random neighbor $u$.  Then the
\emph{heat kernel pagerank} is defined in terms of a preference (row) vector $f$
as $\hkpr = f H_t$.  When $f$, as a row vector, is some probability distribution
over the vertices, the following formulation is useful for our Monte Carlo-based
approximation algorithm:
\begin{equation}\label{eq:hkpr}
\hkpr = f H_t = \sum_{k=0}^{\infty} e^{-t}\frac{t^k}{k!} fP^k.
\end{equation}

In this paper, we consider preference vectors $f = \chi_s$ with all probability
on a single vertex $s$, called the \emph{seed}, and zero probability elsewhere.
This is a common starting distribution for the PageRank vector, as well,
commonly referred to as a \emph{personalized PageRank} (or PPR) vector.  We will
adapt similar terminology and refer to the vector $\phkpr := \rho_{t, \chi_s}$
as the \emph{personalized heat kernel pagerank vector for $s$}, or simply PHKPR.

\subsubsection{Cheeger ratio} For a non-empty subset $S \subset V$ of vertices
in a graph, define the \emph{volume} to be $\vol(S) = \sum_{v\in S} d_v$, where
$d_v$ is the degree of vertex $v$.  The \emph{Cheeger ratio} of a set $S$ is
defined as $\Phi(S) = \frac{|E(S, \bar{S})|}{\min\{\vol(S), \vol(\bar{S})\}}$,
where we use $\bar{S}$ here to denote the set $V \setminus S$, and $E(S,
\bar{S})$ is the set of edges with one endpoint in $S$ and the other in
$\bar{S}$.  The Cheeger ratio of a graph, then, is the minumum Cheeger ratio
over all sets in the graph, $\Phi(G) = \min_{S\subset V} \Phi(S)$.  The Cheeger
ratio provides a quantitative measure concerning graph clusters and is related
to the expansion and spectral gap of a graph~\cite{ch0}.

\subsubsection{Local cluster and sparse cut} The sparse cut problem is to
approximate the Cheeger ratio $\Phi(G)$ of the graph.  This is typically done by
finding a set of vertices whose Cheeger ratio is close to $\Phi(G)$-- that is, a
set which approximates the sparsest cut in the graph.  For the local clustering
problem, however, we are concerned with finding a set with small Cheeger ratio
within a specified subset of vertices.  Alternatively, one can view this as a
sparse cut problem on an induced subgraph.  This Cheeger ratio is sometimes
called a \emph{local Cheeger ratio} with respect to the specified subset.

A local clustering algorithm promises the following.  Given a set $S$ of Cheeger
ratio $\phi$, many vertices in $S$ may serve as seeds for a sweep which finds a
set of Cheeger ratio close to $\phi$.


\subsection{Our Results}
In this work we give a distributed algorithm which computes a local cluster of
Cheeger ratio $O(\sqrt{\phi})$ with high probability, while the optimal local
cluster has Cheeger ratio $\phi$.  Our algorithm finishes in $O\left(
\localclustercomplexity \right)$ rounds in the CONGEST model
(Theorem~\ref{thm:localcluster}) where $\epsilon$ is an error bound.  Further,
if $\phi$ is known, we show how to compute a local cluster in $O\left( \Kparam +
\frac{1}{\epsilon}\right)$ rounds (Theorem~\ref{thm:localclusterphi}). Our
algorithm is an improvement of previous local clustering algorithms by
eliminating a log factor in the performance guarantee.  Further, its running
time improves upon algorithms using standard and PageRank random walks.  In
particular, given the Cheeger ratio of an optimal local cluster, our algorithm
runs in time only dependent upon the approximation error, $\epsilon$, and is
entirely independent of the input graph.  The algorithms and accompanying
analysis are given in Section~\ref{sec:localcluster}.

Similar to existing local clustering algorithms, our algorithm uses a variation
of random walks
to compute a local cluster.  However, rather than a standard random
walk~\cite{st:graphpartitioning:stoc04} or a PageRank random walk with reset
probabilities~\cite{acl:prgraphpartition:focs06}, we use the \emph{heat kernel
random walk} (see Section~\ref{sec:distributedphkpr}).

We remark that in the analysis of random walks, the usual notion of
approximation is total variation distance or some other vector norm based
distance.  However, in the approximation of PageRank or heat kernel pagerank for
large graphs, the definition of approximation is quite different.  Namely, we
say some vector $\phkprapprox$ is an \emph{$\epsilon$-approximate PHKPR vector}
for $\phkpr$ with a seed vertex $s$ and diffusion parameter $t \in R$ if:
\begin{blockquote}
\begin{enumerate}
\item $(1-\epsilon)\phkpr(v) - \epsilon \leq \phkprapprox(v) \leq
(1+\epsilon)\phkpr(v)$, and
\item for each node $v$ with $\phkprapprox(v)=0$, it must be that
$\phkpr(v) \leq \epsilon$.
\end{enumerate}
\end{blockquote}

With the above definition of approximation, we here define the heat kernel
pagerank approximation problem (or the \emph{PHKPR problem} in short): given a
vertex $s$ of a graph and a diffusion parameter $t\in \R$, compute values
$\phkprapprox(v)$ for vertices $v$.
We give a distributed algorithm which solves the PHKPR problem and finishes
after only $O\left( \Kparam \right)$ rounds of communication
(Theorem~\ref{thm:phkprtime}).

We extend our results to distributed $k$-machine model and show the existence of
an algorithm which computes a local cluster over $k$ machines in
$\tilde{O}\left( \frac{\log(\epsilon^{-1})}{\epsilon^3 k^2
\log\log(\epsilon^{-1})} + \frac{1}{\epsilon k^2} + \left(
\frac{\log(\epsilon^{-1})}{k \log\log(\epsilon^{-1})} +
\frac{1}{k\epsilon}\right)\max\left\{ \frac{1}{\epsilon^{3}}, \Delta \right\}
\right)$ rounds, where $\Delta$ is the maximum degree in the graph, with high
probability (Theorem~\ref{thm:localclusterkmachine}).  We note that when hiding
polylogarithmic factors, this time does not depend on the size $n$ of the graph.
We compare this to an algorithm for computing a local cluster with PageRank
which will require $\tilde{O}\left( \frac{\frac{1}{\alpha} + n}{k^2} +
\left(\frac{1}{\alpha k} + \frac{n}{k} \right)\max\{\frac{1}{\epsilon}, \Delta\}
\right)$ rounds with high probability, which is linear in $n$.  These results
are given in Section~\ref{sec:kmachine}.

We briefly note here that local clustering algorithms can easily be extended to
sparse cut algorithms.  Namely, one can sample a number of random nodes in the
network and perform the local clustering algorithm from each.  One node in the
network can store the Cheeger ratios output by each run of the algorithm and
simply return the minimal Cheeger ratio as the value of the sparsest cut in the
network.  In~\cite{st:graphpartitioning:stoc04,acl:prgraphpartition:focs06},
$O(\frac{n}{\sigma}\log n)$ nodes are enough to compute a sparsest cut with high
probability, where $\sigma$ is the size of the cut set.

\section{Fast Distributed Heat Kernel Pagerank
Computation}\label{sec:distributedphkpr} The idea of the algorithm is to launch
a number of random walks from the seed node in parallel, and compute the
fraction of random walks which end at a node $u$ as an estimate of the PHKPR
values $\phkpr(u)$.  Recall the definition of personalized heat kernel pagerank
from (\ref{eq:hkpr}), $\phkpr = \sum_{k=0}^{\infty} e^{-t}\frac{t^k}{k!} ~\chi_s
P^k$.  Then the values of this vector are exactly the stationary distribution of
a \emph{heat kernel random walk}: with probability $p_k = e^{-t}\frac{t^k}{k!}$,
take $k$ random walk steps according to the standard random walk transition
probabilities $P$ (see Section~\ref{sec:hkpr}).

To be specific, the seed node $s$ initializes $r$ tokens, each of which holds a
random variable $k$ corresponding to the length of its random walk.  Then, in
rounds, the tokens are passed to random neighbors with a count incrementor until
the count reaches $k$.  At the end of the parallel random walks, each node
holding tokens outputs the number of tokens it holds divided by $r$ as an
estimate for its PHKPR value.  Algorithm~\ref{alg:phkpr} describes the full
procedure.

\begin{algorithm}
\floatname{algorithm}{Algorithm}
\caption{\estphkpr}
\label{alg:phkpr}
\textbf{input:} a network modeled by a graph $G$, a seed node $s$, a diffusion
parameter $t$, an error bound $\epsilon$\\
\textbf{output:} estimates $\phkprapprox(v)$ of PHKPR values for nodes $v$ in the network\\

\begin{algorithmic}[1]
    \State seed node $s$ generates $r = \frac{16}{\epsilon^3} \log n$ tokens
    $t_i$
    \State $K \gets c\cdot\Kparam$ for any choice of $c \geq 1$
    \State each token $t_i$ does the following: pick a value $k$ with
    probability $p_k = e^{-t}\frac{t^k}{k!}$, then hold the counter value $k_i \gets
    \min\{k, K\}$
    \For{iterations $j=1 \ldots K$}
        \State every node $v$ performs the following in parallel:
        \For{every token $t_i$ node $v$ currently holds}
            \If{$k_i == j$}
                \State hold on to this token for the duration of the iterations
            \Else
                \State send $t_i$ to a random neighbor
            \EndIf
        \EndFor
    \EndFor
    \State let $C_v$ be the number of tokens node $v$ currently holds
    \State each node with $C_v > 0$ returns $C_v/r$ as an estimate for its PHKPR
    value $\phkpr(v)$
\end{algorithmic}
\end{algorithm}

The algorithm is based on that given in~\cite{cs:hkprcluster:iwoca14} in a
static setting.  Theorem 1 of~\cite{cs:hkprcluster:iwoca14} states that an
$\epsilon$-approximate PHKPR vector can be computed with the above procedure by
setting $r = \frac{16}{\epsilon^3}\log n$.  Further, the approximation guarantee
holds when limiting the maximum length of random walks to $K = O\left( \Kparam
\right)$, so that each token is passed for $\max\{k, K\}$ rounds, where $k$ is
drawn with probability $p_k$ as described above.  In the static setting, this
limit keeps the running time down.

In contrast, the distributed algorithm \estphkpr~takes advantage of
decentralized control to take multiple random walk steps via multiple edges at a
time.  That is, through parallel execution, the running time depends only on the
length of random walks, whereas when running the random walks in serial, as
in~\cite{cs:hkprcluster:iwoca14}, the running time must also include the number
of random walks performed.  Thus, keeping $K$ small is critical in keeping the
number of rounds low, and is the key to the efficiency of our local clustering
algorithms.

The correctness of the algorithm follows directly from Theorem 1
in~\cite{cs:hkprcluster:iwoca14}, and is stated here without proof.  The authors
additionally give empirical evidence of the correctness of the algorithm with
parameters $r = \frac{16}{\epsilon^3} \log n$ and $K =
\frac{2\log(\epsilon^{-1})}{\log\log(\epsilon^{-1})}$ in an extended version of
the paper~\cite{chung2015computing}.  They specifically demonstrate that the
ranking of nodes obtained with an $\epsilon$-approximate PHKPR vector computed
this way is very close to the ranking obtained with an exact vector.

\begin{theorem}
For any network $G$, any seed node $s \in V$, and any error bound $0 <
\epsilon < 1$, the distributed algorithm \estphkpr~outputs an
$\epsilon$-approximate PHKPR vector with probability at least $1-\epsilon$.
\end{theorem}

The correctness of the algorithm holds for any choice of $t$, and in fact we use
a particular value of $t$ in our local clustering algorithm (see
Section~\ref{sec:localcluster}).  Regardless, it is clear that the running time
is independent of any choice of $t$.  In fact, we demonstrate in the proof of
Theorem~\ref{thm:phkprtime} that it is independent of $n$ as well.

\begin{theorem}\label{thm:phkprtime}
For any network $G$, any seed node $s \in V$, and any error bound $0 < \epsilon
< 1$, the distributed algorithm \estphkpr~finishes in $O\left(\Kparam\right)$
rounds.
\end{theorem}

\begin{proof}
We show that there is no congestion in the network during any round of the
algorithm; i.e., there are never more than $O(\log n)$ bits sent over any
edge in any iteration of the random walk process.  The proof then follows since
each step of the random walk requires only one round of computation.

In any run of the algorithm, $\frac{16}{\epsilon^3}\log n$ tokens are created,
each holding a message $k_i$ corresponding to a random walk length.  The token
contains no other information.  In particular, no node IDs are transmitted
through the tokens.  Therefore passing a token involves sending a message of
constant size in any iteration of the algorithm.  In the worst case, every token
is transmitted through a single edge in a single iteration of the algorithm.
However, this is still only $O(\rparam)$ bits, and so meets the constraints of
the model.  Namely, even the worst case of sending every token over one edge can
be done with a single round of communication.  Therefore any random walk step
requires only one round of communication, and by construction at most
$O\left(\Kparam\right)$ random walk steps are performed in the algorithm.
\end{proof}

\section{Distributed Local Cluster Detection}\label{sec:localcluster}
In this section we present a fast, distributed algorithm for the local
clustering problem.  The backbone of the algorithm involves investigating sets
of nodes which accumulate in decreasing order of their
$\phkprapprox(v)/d_v$ values.  The process is efficient and requires at
most one linear scan of the nodes in the network (we actually show that the
process can be much faster).

We describe the algorithm presently.  Let $p$ be any function over the nodes of
the graph, and let $\pi$ be the ordering of the nodes in decreasing order of
$p(v)/d_v$.  Then the majority of the work of the algorithm is investigating
sufficiently many of the $n-1$ cuts $(S_j, \bar{S}_j)$ given by the first $j$
nodes in the ordering and the last $n-j$ nodes in the ordering, respectively,
for $j = 1, \ldots, n-1$.  However, by ``sufficiently many'' we indicate that we
may stop investigating the cut sets when either the volume or the size of the
set $S_j$ is large.  Assume this point is after $j = \b{j}$.  Then we choose the
cut set that yields the minimum Cheeger ratio among the $\b{j}$ possible cut sets.
We call this process a \emph{sweep}.  As such, our local clustering algorithm is
a \emph{sweep} of a PHKPR distribution vector.

\begin{algorithm}
\floatname{algorithm}{Algorithm}
\caption{\alglocalcluster}
\label{alg:localcluster}
\textbf{input:} a network modeled by a graph $G$, a seed node $s$, a target
cluster size $\sigma$, a target cluster volume $\varsigma$, an optimal Cheeger
ratio $\phi$, an error bound $\epsilon$\\
\textbf{output:} a set of nodes $S$ with $\Phi(S) \in O(\sqrt{\phi})$\\

\begin{algorithmic}[1]
    \State $t \gets \tparam$
    \State compute PHKPR values $\phkprapprox(v)$ with \estphkpr($G, s, t, \epsilon)$\label{line:phkpr}
    \State every node $v$ with a non-zero PHKPR value estimate sends
    $\phkprapprox(v)/d_v$ to every other node with a non-zero PHKPR value
    estimate\label{line:phase11}
    \Comment{\emph{\textbf{Phase 1}}}
    \State let $\pi$ be the ordering of nodes in decreasing order of
    $\phkprapprox(v)/d_v$\label{line:phase12}
    \Comment{\emph{\textbf{Phase 1}}}
    \State compute Cheeger ratios of each of the cut sets
    with a call of the \textbf{Distributed Sweep Algorithm} and output the cut
    set of minimum Cheeger ratio\label{line:phase2}
    \Comment{\emph{\textbf{Phase 2}}}
\end{algorithmic}
\end{algorithm}

In the static setting, this process will take $O(n \log n)$ time in general.
The authors in~\cite{dassarma:sparsecut:dc} give a distributed sweep algorithm
that finishes in $O(n)$ rounds.  We improve the analysis
of~\cite{dassarma:sparsecut:dc} using a PHKPR vector.  The running time of our
sweep algorithm is given in Lemma~\ref{lemma:sweep}.

The sweep involves two phases.  In Phase 1, the goal is for each node to know
its place in the ordering $\pi$.  Each node can compute their own
$\phkprapprox(v)/d_v$ value locally, and we use $O(\frac{1}{\epsilon})$ rounds
to ensure each node knows the $\pi$ values of all other nodes (see the proof of
Lemma~\ref{lemma:sweep}).  In Phase 2, we use the decentralized sweep
of~\cite{dassarma:sparsecut:dc} described presently:

\paragraph{\textbf{Distributed Sweep Algorithm.}} Let $N$
denote the number of nodes with a non-zero estimated PHKPR value after running
the algorithm~\estphkpr.  Assume each node knows its position in ordering $\pi$
after Phase 1.  We will refer to nodes by their place in the ordering.  Define
$S_j$ to be the cut set of the first $j$ nodes in the ordering.  Then computing
the Cheeger ratio of each cut set $S_j$ involves a computation of the volume of
the set as well as $|E(S_j, \bar{S}_j)|$.  Define the following:
\begin{itemize}
\item $L_j^{\pi}$ is the number of neighbors of node $j$ in $S_{j-1}$, and
\item $R_j^{\pi}$ is the number of neighbors of node $j$ in $\bar{S}_j$.
\end{itemize}
Then the Cheeger ratio of each cut set can be computed locally by:
\begin{align}
&\circ|E(S_j, \bar{S}_j)| = |E(S_{j-1}, \bar{S}_{j-1})| - L_j^{\pi} +
R_j^{\pi}, \mbox{ with } |E(S_1, \bar{S}_1)| = d_1\label{eq:recursiveCheeg1}\\
&\circ\vol(S_j) = \vol(S_{j-1}) + L_j^{\pi} + R_j^{\pi}, \mbox{ with } vol(S_1) =
d_1.\label{eq:recursiveCheeg2}
\end{align}

We now show that a sweep can be performed in $O(N)$ rounds.  Each node knows the
IDS of its neighbors and after Phase 1 each node knows the place of every other
node in the ordering $\pi$.  Therefore, each node can compute locally if a
neighbor is in $S_{j-1}$ or $\bar{S}_j$, and so $L_j^{\pi}$ and $R_j^{\pi}$ can
be computed locally for each node $j$.  Each node can then prepare an $O(\log
n)$-bit message of the form $(\mbox{ID}, L_j^{\pi}, R_j^{\pi})$.  Each of the
$N$ messages of this form can then be sent to the first node in the ordering
using the upcasting algorithm (described in the proof of
Lemma~\ref{lemma:sweep}) using the $\pi$ ordering as node rank.  We note that the
$N$ nodes in the ordering are necessarily in a connected component of the
network, and so the upcasting procedure can be performed in $O(N)$ rounds.
Finally, once the first node in the ordering is in possession of the ordering
$\pi$, and the values of $(L_j^{\pi}, R_j^{\pi})$ for every node in the
ordering, it may iteratively compute $\Phi(S_j)$ locally using the rules
(\ref{eq:recursiveCheeg1}) and (\ref{eq:recursiveCheeg2}).  Thus, this node can
output the minimum Cheeger ratio $\phi^*$ as well as the $j^*$ such that
$\Phi(S_{j^*}) = \phi^*$ after $O(N)$ rounds.

\begin{lemma}\label{lemma:sweep}
Performing Phases 1 and 2 of a distributed sweep takes
$O(\frac{1}{\epsilon})$ rounds.
\end{lemma}

\begin{proof}
First we describe how to send $N$ $O(\log n)$-sized messages to a single node in
$O(N)$ rounds of communication.  For this we can use the upcasting algorithm
of~\cite{peleg:congest} (as described in~\cite{dassarma:sparsecut:dc}).  We
first construct a priority BFS tree of the $N$ nodes with non-zero PHKPR value.
We emphasize again that these nodes are necessarily in a connected component of
the network, and it is shown in~\cite{peleg:congest} that such a BFS tree can be
constructed in $O(N)$ time.  Each node in the tree then upcasts its message to
the root node through the edges of the tree.

In Phase 1, the nodes need to be sorted according to their (non-zero) $\pi$
values.  In this case, the nodes use their $\phkprapprox(v)/d_v$ value as their
rank so that the node with the highest $\phkprapprox(v)/d_v$ value is the root
of the tree.  Then each node upcasts its $\phkprapprox(v)/d_v$ value to the root
through the edges of the tree.  The root node locally sorts these values and
then floods all the $\pi$ values to the nodes through tree edges.  The upcast
and flooding process take $O(N)$ rounds to reach each of the nodes in the tree.

Phase 2 consists of the \textbf{Distributed Sweep Algorithm}, where the first
node in the ordering computes the Cheeger ratio of $S_j$ for each node $j$ in
the ordering.  In order to send each of the $(\mbox{ID}, L_j^{\pi}, R_j^{\pi})$
messages to the first node of the ordering we again upcast through the edges of
a priority BFS tree, however in this round we use $\pi$ values as node rank.
The root node is then able to locally compute Cheeger ratios and output the
cutset of minimum Cheeger ratio after $O(N)$ rounds for upcasting.

Thus Phase 1 requires $O(N)$ rounds for upcasting and flooding values.  Phase 2
requires $O(N)$ rounds for upcasting values necessary for locally computing
Cheeger ratios.  Since we compute an $\epsilon$-approximate PHKPR vector as our
distribution, we know that $N$ is no more than $O(\frac{1}{\epsilon})$.  This is
because we assume $\sum_{v\in V}\phkpr(v) = 1$, and so no more than
$\frac{1}{\epsilon}$ vertices can have values at least $\epsilon$.  Thus the
full sweep takes $O(\frac{1}{\epsilon})$ rounds.
\end{proof}

We note here that the time required for the sweep may be reduced if there are
size or volume restraints for the cut set.  In this case, an alternative
distributed sweep algorithm may be utilized.  As usual, we refer to each node by
their place in the ordering $\pi$.  Node 1 begins the sweep by sending
$\vol(S_1), |E(S_1, \bar{S}_1)|$ to node 2.  Then nodes $j = 2, \ldots, N$
iteratively compute $\Phi(S_j)$ using the values of $\vol(S_{j-1}), |E(S_{j-1},
\bar{S}_{j-1})|, L_j^{\pi}$ and $R_j^{\pi}$, and then subsequently sending
$\vol(S_j), |E(S_j, \bar{S}_j)|$ to the next node $j+1$ in the ordering.
Additionally, each node can send the minimum Cheeger ratio $\phi^*$ computed
thus far as well as the $j^*$ such that $\Phi(S_{j^*}) = \phi^*$.  Thus the last
node in the ordering can output $S_{j^*}$.  In this algorithm, each iteration
$j$ will require $d$ rounds of communication, were $d$ is the shortest path
distance between nodes $j-1$ and $j$.  However, no two nodes will ever be at a
distance of greater than $O(\Kparam)$ steps by construction.  In this way, the
first $j$ Cheeger ratios can be computed in $O(j(\Kparam)$ rounds.

If size or volume restraints are placed on the cluster, we may stop the sweep at
node $\b{j}$ when the size or volume of $S_{\b{j}}$ is greater than some
specified value.  We output the set $S_{j^*}$ for that iteration, and this
process requires $O(\b{j}(\Kparam))$ rounds.

\bigskip
The algorithm \alglocalcluster~(Algorithm~\ref{alg:localcluster}) is a complete
description of our distributed local clustering algorithm.  The correctness of
the algorithm follows directly from~\cite{cs:hkprcluster:iwoca14} and we omit
the proof here.

\begin{theorem}
For any network $G$, suppose there is a set of Cheeger ratio $\phi$.  Then at
least half of the vertices in $S$ can serve as the seed $s$ so that for any
error bound $0 < \epsilon < 1$,  the algorithm \alglocalcluster~will find a set
of Cheeger ratio $O(\sqrt{\phi})$ with probability at least $1-\epsilon$.
\end{theorem}


\begin{theorem}\label{thm:localclusterphi}
For any network $G$, any seed node $s\in V$, and any error bound $0 < \epsilon <
1$, the algorithm \alglocalcluster~finishes in $O\left(\Kparam +
\frac{1}{\epsilon}\right)$ rounds.
\end{theorem}

\begin{proof}
The only distributed computations are those for computing approximate PHKPR
values (line~\ref{line:phkpr}) and Phase 1 (lines~\ref{line:phase11}
and~\ref{line:phase12}) and Phase 2 (line~\ref{line:phase2}) of the distributed
sweep.  Computing PHKPR values takes $O\left( \Kparam \right)$ rounds by
Theorem~\ref{thm:phkprtime}, and Phases 1 and 2 together take
$O(\frac{1}{\epsilon})$ rounds by Lemma~\ref{lemma:sweep}.  Thus the running
time follows.
\end{proof}

One possible concern with the algorithm \alglocalcluster~is that one cannot
guarantee knowing the value of $\phi$ ahead of time for any particular node $s$.
Therefore a true local clustering algorithm should be able to proceed without
this information.  This can be achieved by ``testing'' a few values of $\phi$
(and fixing some reasonable values for $\sigma$ and $\varsigma$).  Namely, begin
with $\phi = 1/2$ and run the algorithm above.  If the output cut set $S$
satisfies $\Phi(S) \in O(\sqrt{\phi})$, we are done.  If not, halve the value of
$\phi$ and continue.  There are $O(\log n)$ such guesses, and we have arrived at
the following.

\begin{theorem}\label{thm:localcluster}
For any network $G$, any node $s$, and any error bound $0 < \epsilon < 1$,
there is a distributed algorithm that computes a set $S$ with Cheeger ratio
within a quadratic of the optimal which finishes in $O\left(
\localclustercomplexity \right)$ rounds.
\end{theorem}

In particular, when ignoring polylogarithmic factors, the running time is
$\tilde{O}\left(\Kparam + \frac{1}{\epsilon}\right)$.

\section{Computing Local Clusters in the $k$-Machine Model}\label{sec:kmachine}
In this section we consider a graph on $n$ vertices which is distributed across
$k$ nodes in a computer network.  This is the $k$-machine model introduced
in Section~\ref{sec:models}.

In the $k$-machine model, we consider a network of $k > 1$ distinct machines
that are pairwise interconnected by bidirectional point-to-point communication
links.  Each machine executes an instance of a distributed algorithm.  The
computation advances in rounds where, in each round, machines can exchange
messages through their communication links.  We again assume that each link has
a bandwidth of $O(\log n)$ meaning that $O(\log n)$ bits may be transmitted
through a link in any round.  We also assume no shared memory and no other means
of communication between nodes.  When we say an algorithm solves a problem in
$x$ rounds, we mean that $x$ is the maximum number of rounds until termination
of the algorithm, over all $n$-node, $m$-edge graphs $G$.

In this model we are solving massive graph problems in which the vertices of the
graph are distributed among the $k$ machines.  We assume $n \geq k$ (typically
$n \gg k$).  Initially the entire graph is not known by a single machine but
rather partitioned among the $k$ machines in a ``balanced'' fashion so that the
nodes and/or edges are partitioned approximately evenly among the machines.
There are several ways of partitioning vertices, and we will consider a random
partition, where vertices and incident edges are randomly assigned to machines.
Formally, each vertex $v$ of $G$ is assigned independently and randomly to one
of the $k$ machines, which we call the home machine of $v$.  The home machine of
$v$ thereafter knows the ID of $v$ as well as the IDs and home machines of
neighbors of $v$.

In the remainder of this section we prove the existence of efficient algorithms for
solving the PHKPR and local cluster problems in the $k$-machine model.  Our main
tool is the Conversion Theorem of~\cite{klauck2015distributed}.

Define $M$ as the \emph{message complexity}, the worst case number of messages
sent in total during a run of the algorithm.  Also define $C$ as the
\emph{communication degree complexity}, or the maximum number of messages sent
or received by any node in any round of the algorithm.  Then we use as a key
tool the Conversion Theorem as restated below.

\begin{theorem}[Conversion Theorem~\cite{klauck2015distributed}]
Suppose there is an algorithm $A_C$ that solves problem $P$ in the CONGEST model
for any $n$-node graph $G$ with probability at least $1-\epsilon$ in time
$T_C(n)$.  Further, let $A_C$ use message complexity $M$ and communication
degree complexity $C$.  Then there exists an algorithm $A_k$ that solves $P$ for
any $n$-node graph $G$ with probability at least $1-\epsilon$ in the $k$-machine
model in $\tilde{O}\left( \frac{M}{k^2} + \frac{T_C(n)C}{k}\right)$ rounds with
high probability.
\end{theorem}

In the forthcoming theorems, by ``high probability'' we mean with probability at
least $1- 1/n$.

We note that the proof of the Conversion Theorem is constructive, describing
precisely how an algorithm $A_k$ in the $k$-machine model simulates the
algorithm $A_C$ in the CONGEST model.  We omit the simulation here but encourage
the reader to refer to the proof for implementation details.

By Theorem~\ref{thm:phkprtime}, we know that PHKPR values can be estimated with
$\epsilon$-accuracy in $O\left( \Kparam \right)$ rounds.  A total of $O\left(
\rparam \right)$ messages are generated and propogated for at most $O\left(
\Kparam \right)$ random walk steps, for a total of $O\left( \rK \right)$
messages sent during a run of the algorithm.  In the first random walk step,
each of the $O\left( \rparam \right)$ messages may be passed to a neighbor of
the seed node, so the message complexity is $O\left( \rparam \right)$.
Therefore we arrive at the following.

\begin{theorem}
There exists an algorithm that solves the PHKPR problem for any $n$-node graph
in the $k$-machine model with probability at least $1-\epsilon$ and runs in
$\tilde{O}\left( \frac{\log(\epsilon^{-1})}{\epsilon^3 k
\log\log(\epsilon^{-1})} (\frac{1}{k} + 1) \right)$ rounds with high
probability.
\end{theorem}

By Theorem~\ref{thm:localcluster}, a local cluster about any seed node can be
computed in $O\left( \localclustercomplexity \right)$ rounds.  The message
complexity for the PHKPR phase is $O\left( \left( \rK \right)\log n \right)$ and
for the sweep phase is $O\left( \frac{1}{\epsilon} \log n \right)$, for a total
message complexity of $O\left( \frac{\log(\epsilon^{-1}) \log^2 n}{\epsilon^3
\log\log(\epsilon^{-1})} + \frac{1}{\epsilon} \log n \right)$.  The
communication degree complexity is $O\left( \rparam \right)$ for the PHKPR phase
(as above), and $O(\Delta)$, where $\Delta$ is the maximum degree in the graph,
for the sweep phase.  Thus the communication degree complexity for the algorithm
is the maximum of these two.  We therefore have the following result for the
$k$-machine model.

\begin{theorem}\label{thm:localclusterkmachine}
There exists an algorithm that computes a local cluster for any $n$-node graph
in the $k$-machine model with probability at least $1-\epsilon$ and runs in
$\tilde{O}\left( \frac{\log(\epsilon^{-1})}{\epsilon^3 k^2
\log\log(\epsilon^{-1})} + \frac{1}{\epsilon k^2} + \left(
\frac{\log(\epsilon^{-1})}{k \log\log(\epsilon^{-1})} +
\frac{1}{k\epsilon}\right)\max\left\{ \frac{1}{\epsilon^{3}}, \Delta \right\}
\right)$ rounds, where $\Delta$ is the maximum degree in the graph, with high
probability.
\end{theorem}

\section{Acknowledgements}
The authors would like to warmly thank Yiannis Koutis for discussion and for
suggesting the problem of finding efficient distributed algorithms, as well as
the anonymous reviewers for their suggestions for improving the paper.

{\footnotesize
\bibliographystyle{splncs03}
\bibliography{distributed,compute}

\begin{thebibliography}{10}
\providecommand{\url}[1]{\texttt{#1}}
\providecommand{\urlprefix}{URL }

\bibitem{acl:prgraphpartition:focs06}
Andersen, R., Chung, F., Lang, K.: Local graph partitioning using pagerank
  vectors. In: FOCS. pp. 475--486. IEEE (2006)

\bibitem{ap:evolving:09}
Andersen, R., Peres, Y.: Finding sparse cuts locally using evolving sets. In:
  STOC. pp. 235--244. ACM (2009)

\bibitem{arora2009expander}
Arora, S., Rao, S., Vazirani, U.: Expander flows, geometric embeddings and
  graph partitioning. JACM  56(2), ~5 (2009)

\bibitem{bp:websearchanatomy:isdn98}
Brin, S., Page, L.: The anatomy of a large-scale hypertextual web search
  engine. Computer Networks and ISDN Systems  30(1),  107--117 (1998)

\bibitem{ch0}
Chung, F.: Spectral Graph Theory. American Mathematical Society (1997)

\bibitem{cs:hkprcluster:iwoca14}
Chung, F., Simpson, O.: Computing heat kernel pagerank and a local clustering
  algorithm. In: IWOCA (2014)

\bibitem{chung2015computing}
Chung, F., Simpson, O.: Computing heat kernel pagerank and a local clustering
  algorithm. arXiv preprint arXiv:1503.03155  (2015)

\bibitem{dassarma:sparsecut:dc}
Das~Sarma, A., Molla, A.R., Pandurangan, G.: Distributed computation of sparse
  cuts via random walks. In: ICDCN. pp. 6:1--6:10 (2015)

\bibitem{sarma2013fast}
Das~Sarma, A., Molla, A.R., Pandurangan, G., Upfal, E.: Fast distributed
  pagerank computation. In: Distributed Computing and Networking, pp. 11--26.
  Springer (2013)

\bibitem{das2013distributed}
Das~Sarma, A., Nanongkai, D., Pandurangan, G., Tetali, P.: Distributed random
  walks. JACM  60(1), ~2 (2013)

\bibitem{dean2004mapreduce}
Dean, J., Ghemawat, S.: Mapreduce: Simplified data processing on large
  clusters. OSDI  (2004)

\bibitem{gt:optimalcluster:12}
Gharan, S.O., Trevisan, L.: Approximating the expansion profile and almost
  optimal local graph clustering. In: FOCS. pp. 187--196. IEEE (2012)

\bibitem{klauck2015distributed}
Klauck, H., Nanongkai, D., Pandurangan, G., Robinson, P.: Distributed
  computation of large-scale graph problems. In: SODA. pp. 391--410. SIAM
  (2015)

\bibitem{kloster2014heat}
Kloster, K., Gleich, D.F.: Heat kernel based community detection. In: ACM
  SIGKDD. pp. 1386--1395. ACM (2014)

\bibitem{lldm:localnetwork:08}
Leskovec, J., Lang, K.J., Dasgupta, A., Mahoney, M.W.: Statistical properties
  of community structure in large social and information networks. In: WWW. pp.
  695--704. ACM (2008)

\bibitem{liao:protein:09}
Liao, C.S., Lu, K., Baym, M., Singh, R., Berger, B.: Isorankn: Spectral methods
  for global alignment of multiple protein networks. Bioinformatics  25(12),
  i253--i258 (2009)

\bibitem{ls:mixingisoperimetric:90}
Lov\'asz, L., Simonovits, M.: The mixing rate of markov chains, an
  isoperimetric inequality, and computing the volume. In: FOCS. pp. 346--354.
  IEEE (1990)

\bibitem{ls:randomwalks:93}
Lov\'asz, L., Simonovits, M.: Random walks in a convex body and an improved
  volume algorithm. Random Structures \& Algorithms  4(4),  359--412 (1993)

\bibitem{yucheng2010graphlab}
Low, Y., Gonzalez, J., Kyrola, A., Bickson, D., Guestrin, C., Hellerstein,
  J.M.: Graphlab: A new framework for parallel machine learning. In: UAI. pp.
  340--349 (2010)

\bibitem{malewicz2010pregel}
Malewicz, G., Austern, M.H., Bik, A.J., Dehnert, J.C., Horn, I., Leiser, N.,
  Czajkowski, G.: Pregel: a system for large-scale graph processing. In: SIGMOD
  International Conference on Management of data. pp. 135--146. ACM (2010)

\bibitem{osv:balsep:11}
Orecchia, L., Sachdeva, S., Vishnoi, N.K.: Approximating the exponential, the
  lanczos method and an $\tilde{O}$(m)-time spectral algorithm for balanced
  separator. In: STOC. pp. 1141--1160. ACM (2012)

\bibitem{pandurangan:congest}
Pandurangan, G., Khan, M.: Theory of communication networks. In: Algorithms and
  Theory of Computation Handbook (2010)

\bibitem{peleg:congest}
Peleg, D.: Distributed computing. SIAM Monographs on Discrete Mathematics and
  Applications  5 (2000)

\bibitem{st:graphpartitioning:stoc04}
Spielman, D.A., Teng, S.H.: Nearly-linear time algorithms for graph
  partitioning, graph sparsification, and solving linear systems. In: STOC. pp.
  81--90. ACM (2004)

\end{thebibliography}
}

\end{document}